\documentclass[theapa,hidelinks]{gCOM2e}

\theoremstyle{plain}

\usepackage{xcolor}
\usepackage{amscd}
\usepackage[shortlabels]{enumitem}
\setlist[enumerate]{itemsep=0mm}

\makeatletter
\newif\if@restonecol
\makeatother

\usepackage[justification=centering, small, labelsep=space]{caption}
\usepackage{array}

\newtheorem{theorem}{Theorem}[section]
\newtheorem{proposition}{Proposition}[section] 

\newtheorem{lemma}{Lemma}[section]

\theoremstyle{remark}
\theoremstyle{definition}

\newtheorem{example}{Example}[section]

\begin{document}
\markboth{Dohan Kim}{International Journal of Computer Mathematics}

\title{Group-theoretical vector space model}

\author{Dohan Kim$^{\ast}$\thanks{$^\ast$ Email: dkim@airesearch.kr
\vspace{6pt}}  \\\vspace{6pt}  {\em{A.I. Research Co., 2537-1 Kyungwon Plaza 201, Sinheung-dong, Sujeong-gu, Seongnam-si, Kyunggi-do, 461-811, South Korea}}\\\vspace{6pt}
}
\maketitle

\begin{abstract}
This paper presents a group-theoretical vector space model (VSM) that extends the VSM with a group action on a vector space of the VSM. We use group and its representation theory to represent a dynamic transformation of information objects, in which each information object is represented by a vector in a vector space of the VSM. Several groups and their matrix representations are employed for representing different kinds of dynamic transformations of information objects used in the VSM. We provide concrete examples of how a dynamic transformation of information objects is performed and discuss algebraic properties involving certain dynamic transformations of information objects used in the VSM.
\end{abstract}


\begin{classcode}
15A03; 15A04; 06B15; 68Q55; 68P20
\end{classcode}

\section{Introduction}
Vectors have been widely used in the field of cognitive science~\cite{Aisbett2001, Roach1978}, machine learning~\cite{Duda2012}, semantics~\cite{Lowe2001}, and information retrieval (IR)~\cite{Singhal2001,Grossman2004,Salton1975}. The vector space model (VSM)~\cite{Salton1975,Silva2004} is a model based on a vector space, which represents information objects (e.g., terms, images, documents, queries, etc.) by vectors in a vector space. Each dimension of a vector space represents a feature of an information object, corresponding to a basis element of a vector space of the VSM~\cite{Lowe2001}. A wide variety of weighting schemes~\cite{Lee1997,Chisholm1999,Sun2004,Berry1999} have been proposed and tested, in which each component of an information-object vector reflects the importance of the corresponding feature of an information-object vector. For the weighted information-object vectors, distance functions are often used to determine how to measure the similarity between information-object vectors~\cite{Wong1984}. One common similarity measure between two information-object vectors is the cosine similarity, measuring the cosine of the angle between two information-object vectors in a vector space of the VSM~\cite{Berry1999}. Besides its intuitive nature, the VSM has also been proven to be effective in IR and relevance ranking~\cite{Melucci2005,Silva2004}. Meanwhile, relevance in IR is often context-dependent as information may evolve with the user, place, and time~\cite{Melucci2005,Melucci2008a,Melucci2008b}, which has not been reflected in the classic, standard VSM~\cite{Salton1975,Singhal2001,Lee1997}. Although the VSM incorporating context and its variants have already been researched~\cite{Melucci2005,Melucci2008a,Melucci2008b}, there is a lack of a systematic approach to representing a dynamic transformation of information objects used in the VSM. Moreover, to the best of our knowledge, theoretical foundations of utilizing group theory for the VSM have not been established. In this paper we present a group-theoretical VSM and discuss properties on several types of dynamic transformations of information-object vectors in a vector space of the VSM. We also show that some properties are invariant to certain dynamic transformations of information-object vectors. The rest of this paper is organized as follows. Section~\ref{sec:VectorSpaceModel} gives a brief overview of the VSM. In Section~\ref{sec:GroupActionVSM} we present our group-theoretical VSM to represent a dynamic transformation of information objects used in the VSM. In Section~\ref{sec:relatedwork} we provide related work and discussion. We give concluding remarks in Section~\ref{sec:Conclusions}. In Appendix we provide the necessary mathematical background on vector spaces, groups, and their representations used in this paper.

\section{Vector Space Model (VSM)}
\label{sec:VectorSpaceModel}
In this section we give a brief overview of the classical vector space model (VSM) used in this paper. %Definitions and results in this subsection are found in~\cite{Kim2010}.
Vector spaces play an important role in cognitive science~\cite{Aisbett2001, Roach1978}, semantics~\cite{Lowe2001}, pattern classification~\cite{Duda2012}, and information retrieval (IR). In particular, they are commonly used in IR, where IR concerns with methods and procedures of searching and obtaining the required information from information resource or corpus~\cite{Silva2004,Latha2008}.
In IR the \emph{Boolean retrieval model}~\cite{Manning2008} poses queries having the form of Boolean expression of terms, in which each query consists of terms combined with Boolean operators, such as AND, OR, and NOT. Each document is considered as a set of words in the Boolean retrieval model~\cite{Manning2008}. However, the Boolean retrieval model has some limitations, such as the lack of similarity measure and a document ranking method~\cite{Singhal2001}. In contrast to the Boolean retrieval model, the VSM has a means to measure the similarity between a query and information-object vectors and to rank information-object vectors according to their similarity scores to the query~\cite{Singhal2001,Manning2008}. 

The VSM  has been formalized as a quadruple $<B, W, S, F>$~\cite{Lowe2001}, where $B$ denotes a set of basis elements of a vector space $V$ of the VSM, $W$ specifies a weight function, $S$ is a similarity measure that maps a pair of information-object vectors to a scalar-valued quantity representing their similarity, and $F$ is a transformation that takes one vector space to another vector space. One of the main purposes of $F$ is to reduce the dimensionality of $V$~\cite{Lowe2001,Tous2006}. $F$ may also be the identity map that transforms $V$ to itself. Note that a vector space of the VSM is often considered as a \emph{feature space}~\cite{Turney2010}. Therefore, a wide variety of \emph{feature weighting} (or \emph{scaling}) schemes~\cite{Liu2004} can be inherited, depending on what kind of a feature space is employed for the VSM.
\\\\
(1) \emph{Basis Elements B}: $B$ is a set of basis elements $b_1,\ldots, b_n$ that determine the dimensionality of a vector space $V$ of the VSM. Each dimension of $V$ represents a feature of an information object. Each information-object vector $v$ is generated by $B$, i.e., $v=\sum_{i=1}^{n}w_ib_i$, where $w_i$'s for $i=1,\ldots, n$ are weights or coefficients. Note that if $B^\prime$ is a set of basis elements $b_1^\prime,\ldots, b_n^\prime$ of $V$, then $v$ can also be generated by $B^\prime$, i.e., $v=\sum_{i=1}^{n}w_i^\prime b_i^\prime$. One basis can be converted into another basis to reflect a contextual change of information-object vectors, in which a context may refer to the time, space, semantic of information objects, and so on~\cite{Melucci2005, Melucci2008a}. It means that a basis of a vector space in the VSM can be constructed to represent a context~\cite{Melucci2008a}. \\\\
(2) \emph{Weight function W}: $W$ is a weight function that maps an information object to its normalized form that is often represented as a coordinate vector. Each component of the coordinate vector represents the weight of the corresponding feature of the information object. $W$ is closely related to feature weighting, which relies on the type of a feature space. If a vector space $V$ of the VSM is given as an $n$-dimensional term space~\cite{Silva2004}, then a query vector $Q$ and a document vector $D_i$ are represented as $Q= (w_{Q,1}, w_{Q,2},\ldots, w_{Q,n})$ and $D_i = (w_{i,1}, w_{i,2},\ldots, w_{i,n})$, respectively. Each term represents each feature of an information object, and each component of an information-object vector represents the importance of a term in a document or query vector~\cite{Lee1997}. Note that $n$ distinct terms are considered in an $n$-dimensional term space. There are a wide variety of ways to determine the weight of a term in a given information object. The simplest approach is the \emph{frequency weighting}~\cite{Liu2004}, in which the weight is simply equal to the frequency of a feature. A common approach to term-weighting is the \emph{tf-idf}~\cite{Lee1997,Manning2008} method, where the weight of a term in a document vector is determined by the local and global factor. The local factor (\emph{term frequency} $tf_{i,j}$) indicates how often term $j$ appears in  document $i$, while the global factor (\emph{inverse document frequency} $idf_j$) indicates how often term $j$ appears in a document collection~\cite{Lee1997}. More specifically, the weight of  term $j$ in document $i$ for the \emph{tf-idf} method is defined as $w_{i,j} := tf_{i,j} \times idf_j = tf_{i,j} \times \displaystyle\mathrm{log}(N/df_j)$, where $N$ is the total number of documents in a document collection and $df_j$ denotes the number of documents (in a document collection) containing term $j$~\cite{Lee1997, Manning2008}. Note that the inverse document frequency ($idf_j:=\displaystyle\mathrm{log}(N/df_j)$) assigns a low value to a term that occurs in a large number of documents, while assigning a high value to a term that occurs in a small number of documents in a document collection~\cite{Lee1997}. The interested reader may also refer to~\cite{Liu2004, Dumais1991, Manning2008} for other term-weighting schemes, such as \emph{Entropy weighting}~\cite{Liu2004} and \emph{Logarithmic weighting}~\cite{Dumais1991}.
\\\\
(3) \emph{Similarity measure S}: $S$ is a similarity measure that maps each pair of information-object vectors to a scalar-valued similarity score. The angle between a pair of information-object vectors can be used as a simple similarity measure between the pair of information-object vectors. Specifically, the cosine of the angle can be used as a numeric similarity measure (i.e., 1.0 for identical vectors while 0.0 for orthogonal vectors). Furthermore, if two information-object vectors in a vector space $\mathbb{Re}^n$ of the VSM are normalized to the unit length, the cosine of the angle between two information-object vectors is simply the inner product of two information-object vectors. Now, the cosine similarity between two information-object vectors $v_1$ and $v_2$ is defined as $sim(v_1, v_2):=(v_{1}\cdot v_{2})/(\|v_1\|\|v_2\|)$, where $v_{1}\cdot v_{2}$ denotes the inner product of information-object vectors $v_{1}$ and $v_{2}$. Therefore, in terms of the cosine similarity measure, the higher the value of $sim(u_i, u_j)$, the more similar information-object vectors $u_i$ and $u_j$ are. Other methods are also available for the similarity measure based on a distance function. The interested reader may refer to~\cite{Zezula2006} for further details.
\\\\
(4) \emph{Transformation F}: $F$ is a transformation\footnote{Since a transformation $F$ is often used for dimensionality reduction, it is distinguished from an (invertible) linear transformation in this paper. Note that an invertible linear transformation (i.e., \emph{isomorphism}~\cite{Dummit2004}) from a vector space $V$ to itself serves as an element of the general linear group $GL(V)$ (see Appendix).} that transforms one vector space $V$ to another vector space $V^\prime$.) The main purpose of $F$ is to reduce the dimensionality of $V$ in such a manner that the dimensionality of $V^\prime$ is smaller than the dimensionality of $V$. The matrix decomposition techniques are often used for dimensionality reduction (i.e., \emph{singular value decomposition}~\cite{Lowe2001} and \emph{QR decomposition}~\cite{Berry1999}). In some cases it is also possible to reduce the dimensionality in the preprocessing steps (e.g., \emph{stop word elimination} and \emph{stemming}~\cite{Singhal2001}). $F$ can also be the identity transformation that maps a vector space $V$ to itself. 
\\

Although the preprocessing and dimensionality reduction steps are often necessary for the VSM, we omit them in this paper. The interested reader may refer to~\cite{Lowe2001, Berry1999, Singhal2001} for further details. Unless otherwise stated,  $B$ denotes a set of basis elements of a given vector space, $W$ \emph{tf-idf}, $S$ cosine similarity, and $F$ denotes the identity map in $<B, W, S, F>$ of the VSM used in this paper. We assume that every vector space of the VSM is finite-dimensional in this paper.

\begin{example}
\label{exam}
This example illustrates how the \emph{tf-idf} weighting method and the cosine similarity measure of the VSM are applied to document ranking, where each document and a query are represented by \emph{a bag of words} (unordered words with duplicates allowed)~\cite{Turney2010}. The bag-of-words model is widely used in a document and image representation~\cite{Filliat2007, Fortuna2005}, spam filtering~\cite{Erdelyi2009}, etc. The following figure shows query $Q$ and three documents $D_1$, $D_2$, and $D_3$, each of which is represented by a bag of words. 

\begin{figure}[h!]
\small
$Q$: $\{term_1 \;\; term_2\}$ \\
$D_1$: $\{term_5  \;\; term_1 \;\; term_1 \;\; term_5\}$\\
$D_2$: $\{term_2 \;\; term_3 \;\; term_3  \;\; term_6  \;\; term_4\}$\\
$D_3$: $\{term_2 \;\; term_1 \;\; term_2\}$
\caption{Bag of words for a query $Q$ and documents $D_1$, $D_2$, and $D_3$.}
\label{fig:QandD}
\end{figure}
There are six distinct terms in Figure~\ref{fig:QandD}. Table~\ref{table:TermWeights} shows term weights for each document and query using the \emph{tf-idf} weighting method ~\cite{Grossman2004,Manning2008}.
\begin{table}[ht]
\centering \small
\caption{Term weights for $Q$, $D_1$, $D_2$, and $D_3$ in Figure~\ref{fig:QandD}.}
\label{table:TermWeights}
\begin{tabular}{|c|c|c|c|c|c|c|c|c|c|c|c|}
\hline
\multicolumn{12}{|l|}{$\textrm{Term weights } w_{i,j}:=tf_{i,j} \times idf_j$ } \\
\multicolumn{12}{|l|}{$(tf_{i,j}:\textrm{term frequency}, \;idf_j:\textrm{inverse document frequency})$} \\ \hline
\multicolumn{12}{|l|}{$\textrm{Total number of documents N=3}, idf_j:=\textrm{log}(N/df_j)$} \\ 
\multicolumn{12}{|l|}{($df_j:\textrm{number of documents containing term}\,j$)} \\ \hline

& \multicolumn{4}{c|}{$tf_{i,j}$} & & &&
\multicolumn{4}{c|}{$w_{i,j}=tf_{i,j} \times idf_j$}\\ \hline 

$\textrm{Terms}$ & $Q$ & $D_1$& $D_2$&$D_3$& $df_j$ & $N/df_j$ & $idf_j$ & $Q$ & $D_1$ & $D_2$ & $D_3$ \\

\hline
$term_1$ & 1 & 2& 0 & 1 & 2 & 3/2 & 0.176 & 0.176 & 0.352 & 0 &0.176\\
\hline
$term_2$ & 1 & 0& 1 & 2 & 2 & 3/2 & 0.176 & 0.176 & 0 & 0.176 & 0.352\\
\hline
$term_3$ & 0 & 0& 2 & 0 & 1 & 3/1 & 0.477 & 0 & 0 & 0.954 &0\\
\hline
$term_4$ & 0 & 0& 1 & 0 & 1 & 3/1 & 0.477 & 0 & 0 & 0.477 & 0\\
\hline
$term_5$ & 0 & 2 & 0 & 0 & 1 & 3/1 & 0.477 & 0 & 0.954 & 0 & 0\\
\hline
$term_6$ & 0 & 0& 1 & 0 & 1 & 3/1 & 0.477 & 0 & 0 & 0.477 & 0\\
\hline

\end{tabular}
\end{table}

Using Table~\ref{table:TermWeights}, we compute the cosine similarity between $Q$ and $D_i$ for $1\leq i \leq 3$. Since $\|D_i\|=\sqrt{\sum_jw_{i,j}^2}$ and $\|Q\|=\sqrt{\sum_jw_{Q,j}^2}$, we have
\\\\
$\|Q\|=\sqrt{0.176^2+0.176^2} \approx \sqrt{0.062} \approx 0.249$,\\
$\|D_1\|=\sqrt{0.352^2 + 0.954^2} \approx \sqrt{1.034} \approx 1.017$,\\
$\|D_2\|=\sqrt{0.176^2+0.954^2+0.477^2+0.477^2} \approx \sqrt{1.396} \approx {1.182}$,\\
$\|D_3\|=\sqrt{0.176^2+0.352^2} \approx \sqrt{0.155} \approx{0.394}$.\\\\
Since $Q \cdot D_i = \sum_jw_{Q,j}w_{i,j}$, we have
\\\\
$Q \cdot D_1 = 0.176 \times 0.352 \approx 0.062$,\\
$Q \cdot D_2 = 0.176 \times 0.176 \approx 0.031$,\\
$Q \cdot D_3 = (0.176 \times 0.176) + (0.176 \times 0.352) \approx 0.093$.\\\\
Now, the cosine similarity measure between query $Q$ and document $D_i$\footnote{By a slight abuse of notation, we use a document (respectively, a query) and its document vector (respectively, query vector) with the same notation in this paper. The distinction is clear from the context.}  for $1\leq i\leq 3$ are computed as follows:\\\\
$sim(Q, D_1)=(Q \cdot D_1)/(\|Q\|\|D_1\|) \approx 0.062/(0.249 \times 1.017) \approx 0.245$,\\
$sim(Q, D_2)=(Q \cdot D_2)/(\|Q\|\|D_2\|) \approx 0.031/(0.249 \times 1.182) \approx 0.105$,\\
$sim(Q, D_3)=(Q \cdot D_3)/(\|Q\|\|D_3\|) \approx 0.093/(0.249 \times 0.394) \approx 0.949$.\\\\
\indent For the given query $Q$, $D_3$ shows the highest rank with $sim(Q, D_3)\approx0.949$, while $D_2$ shows the lowest rank with $sim(Q, D_2)\approx0.105$. Therefore, according to the cosine similarity measure, document $D_3$ is the most similar to query $Q$, while document $D_2$ is the least similar to query $Q$.
\end{example}

\section{Group Actions on a Vector Space of the VSM}
\label{sec:GroupActionVSM}

In this section we use several groups to represent dynamic transformations of information objects. We show that some properties are preserved for certain dynamic transformations, in which those dynamic transformations of information objects are represented by a group action on a vector space of the VSM. For the bag-of-words model, we assume that although the content of an information object can be changed by a dynamic transformation, no term can be introduced during a dynamic transformation of information objects. We first describe how an orthogonal group acts on a vector space $V=\mathbb{Re}^n$ of the VSM.

For each $n$, the set of all $n \times n$ orthogonal matrices with real entries forms a subgroup of $GL(n, \mathbb{Re})$, denoted by $O(n, \mathbb{Re})$, in which a square matrix $M$ is called \emph{orthogonal} if $M^\top=M^{-1}$~\cite{Fulton1991}. A linear transformation $T:\mathbb{Re}^n \rightarrow \mathbb{Re}^n$ is called an \emph{orthogonal transformation} on $\mathbb{Re}^n$ if its transformation matrix in the standard (ordered) basis is an orthogonal matrix with real entries~\cite{Anton1993}. Orthogonal matrices include rotation and permutation matrices~\cite{Anton1993}.
\begin{proposition}
\label{prop:orthogonal}
Let $V=\mathbb{Re}^n$ be a vector space of the VSM. If $O(n, \mathbb{R})$ acts on $V$ by matrix multiplication, then it preserves the cosine similarity between information-object vectors in $V$.
\end{proposition}
\begin{proof}
By the definition of the orthogonal group, if $O(n, \mathbb{R})$ acts on $V$ by matrix multiplication, we have $M v_1 \cdot M v_2 = (Mv_1)^\top Mv_2 = v_1^\top M^\top Mv_2 = v_1^\top v_2 = v_1 \cdot v_2$ for $M \in O(n, \mathbb{Re})$ and $v_1, v_2 \in V$. Since $(Mv)^\top (Mv)=v^\top v$, we have $\|Mv\|=\|v\|$ for $v \in V$.
Therefore, an orthogonal matrix preserves both the inner product and the length of information-object vectors. It follows that for information-object vectors $u, v \in V$ and $g \in O(n, \mathbb{R})$, we have $sim(gu, gv)=(gu \cdot gv)/(\|gu\|\|gv\|)=(u\cdot v)/(\|u\|\|v\|)=sim(u, v)$.
\end{proof}
%\noindent\textcolor{blue}{\textbf{Remarks. }In IR orthogonal transformations are often used in \emph{document clustering}~\cite{Hu2008, Hoenkamp2003}, in which documents in the same cluster are intended to be more similar than documents in different clusters. Let $V=\mathbb{Re}^n$ be a vector space of the VSM for document clusters. By Proposition~\ref{prop:orthogonal}, the cosine similarity between document vectors in each document cluster is invariant  under orthogonal transformations on $V$.}
%\\
\begin{example}\label{example:ex1} (\emph{Householder matrix}~\cite{Householder1958,Bischof1987}) Consider a reflection linear operator $R \in GL(V)$ of a vector space $V=\mathbb{Re}^n$ of the VSM that reflects each information-object vector through the vector hyperplane that is orthogonal to a unit vector $u$. The transformation matrix $[R]$ of the linear operator $R$ with respect to the standard (ordered) basis of $\mathbb{Re}^n$ is called a \emph{Householder matrix}, and is given by $I - 2uu^\top$. Let $H$ denote $[R]$. Since $HH^\top=I$ (see~\cite{Bischof1987}), we have $H \in O(n, \mathbb{R})$. 

For instance, suppose that we have six terms (i.e., $term_1$, $term_2$, $term_3$, $term_4$, $term_5$, and $term_6$) and three documents (i.e., $D_1$, $D_2$, and $D_3$) as shown in Table~\ref{table:TermWeights}. The $6 \times 3$ term-by-document matrix $D$ is denoted as follows:
%\vfill
\[ D=\left( \begin{array}{ccc}
0.352$\text{}$ & 0$\text{ }$ & 0.176\\
0$\text{}$ & 0.176$\text{ }$ & 0.352\\
0$\text{}$ & 0.954$\text{ }$ & 0\\
0$\text{}$ &0.477$\text{ }$ & 0\\
0.954$\text{}$ & 0$\text{ }$&0\\
0$\text{}$ & 0.477$\text{ }$&0\end{array} \right).\]

Each column corresponds to a document $D_j$, while each row corresponds to a term $term_i$. Each element $d_{ij}$ in $D$ represents the term weight of $\text{term}_i$ associated with document $D_j$. Let $u=[-\sqrt{2}/2, \sqrt{2}/2, 0, 0, 0, 0]^\top$ and select the vector hyperplane of $\mathbb{Re}^6$ that is orthogonal to $u$. Then, the Householder matrix $H^\prime$ is computed as follows:

\[ H^\prime=\left( \begin{array}{cccccc}
0$\text{ }$ & 1$\text{ }$ & 0$\text{ }$&0$\text{ }$&0$\text{ }$&0 \\
1$\text{ }$ & 0$\text{ }$ & 0$\text{ }$&0$\text{ }$&0$\text{ }$&0 \\
0$\text{ }$ & 0$\text{ }$ & 1$\text{ }$&0$\text{ }$&0$\text{ }$&0 \\
0$\text{ }$ & 0$\text{ }$ & 0$\text{ }$&1$\text{ }$&0$\text{ }$&0 \\
0$\text{ }$ & 0$\text{ }$ & 0$\text{ }$&0$\text{ }$&1$\text{ }$&0 \\
0$\text{ }$ & 0$\text{ }$ & 0$\text{ }$&0$\text{ }$&0$\text{ }$&1 \\ \end{array} \right).\] 

Now, the transformation of $D$ by $H^\prime$ is computed as follows:
\[ D^\prime=H^\prime D=\left( \begin{array}{ccc}
0$\text{}$ & 0.176$\text{ }$ & 0.352\\
0.352$\text{}$ &0$\text{ }$ & 0.176\\
0$\text{}$ &0.954$\text{ }$ & 0\\
0$\text{}$ &0.477$\text{ }$ & 0 \\
0.954$\text{}$ & 0$\text{ }$&0\\
0$\text{}$ &0.477$\text{ }$&0\end{array} \right).\]

The first column of $D^\prime$ represents the transformation of $D_1$ by $H^\prime$, the second column of $D^\prime$ the transformation of $D_2$ by $H^\prime$, and the third column of $D^\prime$ represents the transformation of $D_3$ by $H^\prime$. It basically replaces $term_1$ with $term_2$, and vice versa\footnote{For a further consideration of a permutation of the basis vectors, consider a symmetric group $S_n$ acting on a vector space $V=\mathbb{Re}^n$. Let $B=\{b_1,\ldots, b_n\}$ be a basis of $V$. Then, $S_n$ acts on $V$ by $g (\sum_i{c_ib_i}) = \sum_i{c_ib_{g(i)}}$ for $g \in S_n$, $c_i \in \mathbb{Re}$, and $\sum_ic_ib_i \in V$. See~\cite{Alperin1995, Dummit2004} for further details.}, in $D_1$, $D_2$, and $D_3$. Since $H^\prime \in O(6, \mathbb{R})$, the cosine similarity among $D_1$, $D_2$, and $D_3$ are preserved among $H^\prime D_1$, $H^\prime D_2$, and $H^\prime D_3$ by Proposition~\ref{prop:orthogonal}.
\end{example}

The set of all $n \times n$ invertible diagonal matrices with real entries forms a subgroup of $GL(n, \mathbb{Re})$~\cite{Rotman1994}, which is denoted by $D(n, \mathbb{Re})$ in this paper. We first describe a \emph{scaling matrix}~\cite{Bretscher1997}. A scaling matrix is a diagonal matrix, in which each element in the main diagonal represents a \emph{scaling factor} $s_i$ for the $i$-th coordinate axis. If $s_i > 1$, it represents a \emph{dilation} transformation in the direction of the $i$-th coordinate axis. If $0<s_i<1$, it represents a \emph{contraction} transformation in the direction of the $i$-th coordinate axis. If $s_i = -1$, it represents a \emph{reflection} transformation in the direction of the $i$-th coordinate axis. Note that if a scaling matrix has no zero in its main diagonal, it is invertible.

We say that a linear operator $T:V\rightarrow V$ is \emph{diagonalizable scaling linear operator} if there exists an (ordered) basis of $V$ with respect to which the transformation matrix of $T$ is an invertible scaling matrix.

\begin{proposition}
\label{prop:nonprojective}
Let $V$ be an $n$-dimensional vector space over $\mathbb{Re}$ of the VSM. If $D(n, \mathbb{Re})$ acts on $V$ by matrix multiplication, $d \in D(n, \mathbb{Re})$ represents a transformation matrix of a diagonalizable scaling linear operator of $V$.
\end{proposition}
\begin{proof}
By the definition of $D(n, \mathbb{Re})$, $d \in D(n, \mathbb{Re})$ is a diagonal matrix. Since $D(n, \mathbb{Re})$ is a subgroup of $GL(n, \mathbb{Re})$, $d \in D(n, \mathbb{Re})$ is invertible. It follows that its determinant is not zero. Therefore, each $d_i$ in the main diagonal of $d$ is not zero and may represent a scaling factor in the direction of the $i$-th coordinate axis. It follows that $d$ is an invertible scaling matrix that represents a transformation matrix of a diagonalizable scaling linear operator of $V$ with respect to a given basis of $V$.
\end{proof}
\begin{example}\label{example:ex2}
Each component of an information-object vector can be varied by change of context (e.g., a time-dependent document collection~\cite{Nunes2011,Elsas2010}). This example shows the systematic way of changing weights using an invertible scaling matrix. The $6 \times 3$ term-by-document matrix $D$ in Example~\ref{example:ex1} was given as:
\[ D=\left( \begin{array}{ccc}
0.352$\text{}$ & 0$\text{ }$ & 0.176\\
0$\text{}$ & 0.176$\text{ }$ & 0.352\\
0$\text{}$ & 0.954$\text{ }$ & 0\\
0$\text{}$ &0.477$\text{ }$ & 0\\
0.954$\text{}$ & 0$\text{ }$&0\\
0$\text{}$ & 0.477$\text{ }$&0\end{array} \right).\] 

Suppose that an invertible scaling matrix $S$ is given below:
\[ S=\left( \begin{array}{cccccc}
2$\text{ }$& 0$\text{ }$ & 0$\text{ }$& 0$\text{ }$ & 0 $\text{ }$& 0\\
0$\text{ }$& 3$\text{ }$ & 0$\text{ }$ & 0$\text{ }$ & 0 $\text{ }$& 0\\
0$\text{ }$& 0$\text{ }$ & 2$\text{ }$ & 0$\text{ }$ & 0 $\text{ }$& 0\\
0$\text{ }$& 0$\text{ }$ & 0$\text{ }$ & 1$\text{ }$ & 0 $\text{ }$& 0\\
0$\text{ }$& 0$\text{ }$ & 0$\text{ }$ & 0$\text{ }$ & 1 $\text{ }$& 0 \\
0$\text{ }$& 0$\text{ }$ & 0$\text{ }$ & 0$\text{ }$ & 0 $\text{ }$& 1\end{array} \right).\]

Then, the transformation of $D$ by $S$ is computed as follows:
\[ D''=SD=\left( \begin{array}{ccc}
0.704$\text{}$ & 0$\text{ }$ & 0.352\\
0$\text{}$ & 0.528$\text{ }$ & 1.056\\
0$\text{}$ & 1.908$\text{ }$ & 0\\
0$\text{}$ &0.477$\text{ }$ & 0\\
0.954$\text{}$ & 0$\text{ }$&0\\
0$\text{}$ & 0.477$\text{ }$&0\end{array} \right).\] 

The first column of $D''$ represents the transformation of $D_1$ by $S$, the second column of $D''$ the transformation of $D_2$ by $S$, and the third column of $D''$ represents the transformation of $D_3$ by $S$. By means of the scaling matrix $S$, the weight of $term_1$ and the weight of $term_3$ in a document collection are multiplied by two, while the weight of $term_2$ in a document collection is multiplied by three. The weight of $term_4$, the weight of $term_5$, and the weight of $term_6$ are invariant under $S$.
\end{example}

\begin{proposition}
\label{prop:dsoperator}
Let $V$ be an $n$-dimensional vector space over $\mathbb{Re}$ of the VSM. A square matrix $s \in GL(n, \mathbb{Re})$ has $n$ linearly independent eigenvectors if and only if it represents a diagonalizable scaling linear operator of $V$.
\end{proposition}
\begin{proof}
$(\Rightarrow)$\\
Assume that a square matrix $s \in GL(n, \mathbb{Re})$ has $n$ linearly independent eigenvectors. Since $s \in GL(n, \mathbb{Re})$ by assumption, $s$ is invertible. It follows that the determinant of $s$ is not zero. Since $s$ has $n$ linearly independent eigenvectors by assumption and similar matrices have the same determinant, $s$ is diagonalizable to an invertible diagonal matrix $s^\prime \in D(n, \mathbb{Re})$ by Theorem~\ref{thm:diagonalizable}\footnote{See Appendix for Theorem A.1--A.4 and Lemma A.1--A.3.} and Lemma~\ref{lem:determinants}. Therefore, by Theorem~\ref{thm:changeofbasis} and Proposition~\ref{prop:nonprojective}, it represents a diagonalizable scaling linear operator of $V$.\\
$(\Leftarrow)$\\
If $s \in GL(n, \mathbb{Re})$ represents a diagonalizable scaling linear operator of $V$, then it is diagonalizable by Theorem~\ref{thm:changeofbasis}. Therefore, $s\in GL(n, \mathbb{Re})$ has $n$ linearly independent eigevenvectors by Theorem~\ref{thm:diagonalizable}.
\end{proof}
\noindent\textbf{Remarks. }The above proof of Proposition~\ref{prop:dsoperator} involves Theorem~\ref{thm:changeofbasis}, which in turn involves a change of basis of a vector space. In ~\cite{Melucci2005,Melucci2008a,Melucci2008b,Mbarek2014a,Mbarek2014b} context is modeled by a basis of a vector space of the VSM. By Proposition~\ref{prop:dsoperator}, a certain type of invertible linear operators of a vector space of the VSM can be simplified to a type of diagonalizable linear operators by means of a change of context if context is modeled by a basis of a vector space of the VSM.\\

Similarly to the above proposition, we have the following lemma by Lemma~\ref{lem:symmetricmatrix}.
\begin{lemma}
\label{lem:sym}
Let $V$ be an $n$-dimensional vector space over $\mathbb{Re}$ of the VSM. If $s \in GL(n, \mathbb{Re})$ is symmetric, it represents a diagonalizable scaling linear operator of $V$.
\end{lemma}

In Example~\ref{example:ex1} we considered a Householder matrix $H$ given by $I - 2uu^\top$, where $u$ is a unit vector orthogonal to the selected vector hyperplane. Since $(I - 2uu^\top)^\top=I^\top-2(u^\top)^\top(u^\top) = I - 2uu^\top$, it is symmetric. Since the determinant of a householder matrix is $-1$~\cite{Householder1958}, we have $H \in GL(n, \mathbb{Re})$. Therefore, by Lemma~\ref{lem:sym}, $H$ may represent a diagonalizable scaling linear operator of $V=\mathbb{Re}^n$.

\begin{example}
Suppose that a feature space $V=\mathbb{Re}^4$ of the VSM has four features ($location_1$, $location_2$, $height$, and $brightness$) and some normalized feature vectors. Let $B=\{e_1, e_2, e_3, e_4\}$ denote the standard (ordered) basis of $\mathbb{Re}^4$. Now, four features in the feature space are interpreted in such a manner that $e_1:=location_1, e_2:=location_2, e_3:=height$, and $e_4:=brightness$. Suppose also that the transformation matrix $[T]_B$ of a linear operator $T:V \rightarrow V$ with respect to $B$ is given as follows.
\[ [T]_B=\left( \begin{array}{cccc}
3$\text{ }$& 1$\text{ }$ & 0$\text{ }$& 0\\
1$\text{ }$& 3$\text{ }$ & 0$\text{ }$ & 0\\
0$\text{ }$& 0$\text{ }$ & 1$\text{ }$ & 0 \\
0$\text{ }$& 0$\text{ }$ & 0$\text{ }$ & 1 \end{array} \right).\]

Since $[T]_B$ is an invertible and symmetric matrix, it is diagonalizable to an invertible diagonal matrix in $D(4, \mathbb{Re})$ by Lemma~\ref{lem:symmetricmatrix} and~\ref{lem:determinants}. In other words, there is a transition matrix $P$ from an ordered basis $B^\prime=\{e_1^\prime, e_2^\prime, e_3^\prime, e_4^\prime\}$ to $B=\{e_1, e_2, e_3, e_4\}$ such that $[T]_{B^\prime} = P^{-1}[T]_BP$ is an invertible diagonal matrix, i.e., $[T]_{B^\prime} \in D(4, \mathbb{Re})$. By using the diagonalization procedure (see~\cite{Anton1993}), we have

\[ P=\left( \begin{array}{cccc}
1/\sqrt{2}$\text{ }$& -1/\sqrt{2}$\text{ }$ & 0$\text{ }$& 0\\
1/\sqrt{2}$\text{ }$& 1/\sqrt{2}$\text{ }$ & 0$\text{ }$ & 0\\
0$\text{ }$& 0$\text{ }$ & 1$\text{ }$ & 0 \\
0$\text{ }$& 0$\text{ }$ & 0$\text{ }$ & 1 \end{array} \right)\;\;\;\;\textrm{and}\;\;\;\;
[T]_{B^\prime}=\left( \begin{array}{cccc}
4$\text{ }$& 0$\text{ }$ & 0$\text{ }$& 0\\
0$\text{ }$& 2$\text{ }$ & 0$\text{ }$ & 0\\
0$\text{ }$& 0$\text{ }$ & 1$\text{ }$ & 0 \\
0$\text{ }$& 0$\text{ }$ & 0$\text{ }$ & 1 \end{array} \right).\] 

It follows that $e_1^\prime=1/\sqrt{2}e_1 - 1/\sqrt{2}e_2$, $e_2^\prime=1/\sqrt{2}e_1 + 1/\sqrt{2}e_2$, $e_3^\prime=e_3$, and $e_4^\prime=e_4$. Therefore, $[T]_{B^\prime}\in D(4, \mathbb{Re})$ is a transformation matrix of a  diagonalizable scaling linear operator of $V$ with respect to $B^\prime$ by Proposition~\ref{prop:nonprojective}. 
\end{example}

%A \emph{transvection matrix} (or \emph{shear matrix})~\cite{Alperin1995,Erickson2011} over $\mathbb{Re}$ is a square matrix formed by replacing one of the off-diagonal entries of $I_n$ (the identity matrix) with $\alpha \in \mathbb{Re}$. 
Let $B(n, \mathbb{Re})$ be the set of all $n \times n$ invertible upper triangular matrices with real entries. $B(n, \mathbb{Re})$ forms a subgroup of $GL(n, \mathbb{Re})$, called the \emph{standard Borel subgroup}~\cite{Alperin1995} of $GL(n, \mathbb{Re})$. 

Let $V=\mathbb{Re}^n$ be a vector space of the VSM and $B=\{e_1, \ldots, e_n\}$ be its fixed standard (ordered) basis. Then, there is an ascending chain of subspaces $\{0\} \subset V_1=\mathbb{Re} \subset V_2=\mathbb{Re}^2 \subset \cdots \subset V_{n-1}=\mathbb{Re}^{n-1} \subset V_n=\mathbb{Re}^n$, in which each $V_i$ for $1 \leq i \leq n$ is spanned by basis elements $e_1, \ldots, e_i$. This ascending chain is called the \emph{standard complete flag}~\cite{Alperin1995} of $V=\mathbb{Re}^n$. The following proposition describes that if information-object vectors in $V=\mathbb{Re}^n$ are transformed by an $n\times n$ invertible upper triangular matrix $m \in B(n, \mathbb{Re})$, it preserves  the standard complete flag of $V=\mathbb{Re}^n$.

\begin{proposition}
\label{prop:uppertriangular}
Let $V=\mathbb{Re}^n$ be a vector space of the VSM and let $\{0\} \subset V_1=\mathbb{Re} \subset V_2=\mathbb{Re}^2 \subset \cdots \subset V_{n-1}=\mathbb{Re}^{n-1} \subset V_n=\mathbb{Re}^n$ be the standard complete flag of $V=\mathbb{Re}^n$. Then, $gV_i=V_i$ for $1 \leq i \leq n$, where $g \in B(n, \mathbb{Re})$ is an $n \times n$ invertible upper triangular matrix with real entries.
\end{proposition}
\begin{proof}
It follows immediately from the fact that the standard Borel subgroup $B(n, \mathbb{Re})$ of $GL(n, \mathbb{Re})$ stabilizes the standard complete flag of $V=\mathbb{Re}^n$ (see~\cite{Alperin1995} for further details).
\end{proof}

\noindent\textbf{Remarks. } In IR high-dimensional information-object vectors in a vector space of the VSM are often projected into a low-dimensional subspace in order to improve computational efficiency~\cite{Turney2010}. Now, consider information-object vectors in a vector space $V_n=\mathbb{Re}^n$ of the VSM and project them into a subspace $V_i$ $(1\leq i \leq n)$ of $V_n=\mathbb{Re}^n$. By Proposition~\ref{prop:uppertriangular}, the projected information-object vectors are transformed and remained in that subspace by a linear operator of $V_n=\mathbb{Re}^n$ if the transformation matrix of the linear operator with respect to the standard (ordered) basis is an invertible upper triangular matrix with real entries.
\\

For instance, if document $D_3$ in Example~\ref{example:ex1} is transformed by $t \in B(6, \mathbb{Re})$, it still resides in the subspace spanned by the basis element corresponding to $term_1$ and the basis element corresponding to $term_2$. This is not the case if document $D_3$ is transformed by, let us say, a $6\times 6$ invertible lower triangular matrix formed by replacing the (3, 1)-entry of the $6 \times 6$ identity matrix with 1.  
%We have considered several subgroups of $GL(n, \mathbb{Re})$. They are closed under matrix multiplication and are invertible by the definition of a group.

We next describe the dual space of a vector space $V$ over $\mathbb{Re}$ of the VSM. Each information-object vector in $V$ may associate with a scalar-valued quantity. For instance, if a bag of words consists of terms involving product or service items in a \emph{recommender system}~\cite{Gkatzioura2013,Wang2012, Linden2003}, each term may associate with a cost (e.g., purchase price). Now, consider the query and documents in Figure~\ref{fig:QandD}, where the vector space of the VSM is $V=\mathbb{Re}^6$. Let $B=\{u_1,\ldots, u_6\}$ be an ordered basis of $\mathbb{Re}^6$. Those six terms are interpreted in such a manner that $u_1:=term_1,\ldots, u_6:=term_6$. Using the frequency weighting scheme, we have $Q=u_1+u_2$, $D_1=2u_1+2u_5$, $D_2=u_2+2u_3+u_4+u_6$, and $D_3=u_1+2u_2$. We now consider the dual space $\widehat V$ of a vector space $V$ of the VSM. Suppose that the costs involving each term are $3$ for $u_1$ and 4, 5, 6, 6 and 7 for $u_2$, $u_3$, $u_4$, $u_5$, and $u_6$, respectively. An important linear functional in the dual space $\widehat V$ of $V$ is $\phi=3\hat u_1 + 4 \hat u_2 +5 \hat u_3 + 6 \hat u_4 + 6 \hat u_5 + 7 \hat u_6$, in which $<\phi, u_1>=3, <\phi, u_2>=4$, and so on. By pairing $\phi$ with a term, the cost of the term is restored. Similarly, by pairing $\phi$ with an information-object vector, the total cost of the information-object vector is obtained. For instance, the total cost of $D_2$ is $<\phi, D_2>=4+ 2\times 5+6+7=27$. The following proposition describes the relationship between the representation $\rho:G \rightarrow GL(V)$ of $G$ and the dual representation $\hat \rho(g)=[\rho(g^{-1})]^\top:\widehat V \rightarrow \widehat V$ of $G$ used in the VSM.

\begin{proposition}
\label{prop:dual}
Let $V$ be a vector space over $\mathbb{Re}$ of the VSM and $\widehat V$ be the dual space of $V$. Let $\rho:G \rightarrow GL(V)$ be a representation of $G$ and let $\hat \rho:G \rightarrow GL(\widehat V)$ be the dual representation of $G$ to $\rho:G \rightarrow GL(V)$ acting on $\widehat V$ given by
\begin{center}
$\hat \rho(g)=[\rho(g^{-1})]^\top:\widehat V \rightarrow \widehat V$.
\end{center}
%\begin{center}
Then, $<\hat \rho(g)(\hat v), \rho(g)(v)>=<\hat v, v>$ for all $g \in G, v \in V$, and $\hat v \in \widehat V$.
%\end{center}
\end{proposition}
\begin{proof}
See Lemma~\ref{lem:dualrep}.
\end{proof}
\begin{example}
Let $V=\mathbb{Re}^2$ be two dimensional vector space with standard (ordered) basis elements $e_1=[1, 0]^\top$ and $e_2=[0, 1]^\top$, and let $\widehat V$ be the dual space of $V$ with ordered basis elements $\hat e_1$ and $\hat e_2$ selected by Theorem~\ref{thm:dualbasis}. Let $u$ be an information-object vector in $V$ and $\psi$ be a linear functional in $\widehat V$ such that $u=e_1+e_2$ and $\psi=2\hat e_1 + 4\hat e_2$. We use the frequency weighting scheme for $u$ and $\psi$. Therefore, the document corresponding to $u$ consists of the term corresponding to $e_1$ and the term corresponding to $e_2$. Similarly, $\psi$ can be interpreted in such a manner that the cost of the term corresponding to $e_1$ is 2 and the cost of the term corresponding to $e_2$ is 4. The total cost of $u$ is obtained by pairing $\psi$ with $u$, that is, $<\psi, u>=<2\hat e_1+4\hat e_2, e_1+e_2>=6$. Let $\rho:D(2, \mathbb{Re})\rightarrow GL(2, \mathbb{Re})$ be a matrix representation of $D(2, \mathbb{Re})$ associated with $V$ such that $\rho(g)=\left(
\begin {matrix}
1 & 0 \\
0 & 2 \\
\end {matrix}
\right)
$ for $g \in D(2, \mathbb{Re})$. Then, $\rho(g)$ transforms $u=e_1+e_2$ into $u^\prime=e_1 + 2e_2$. We see that $u^\prime$ now consists of a single $e_1$ and two ${e_2}$'s. Let $\hat \rho:D(2, \mathbb{Re}) \rightarrow GL(2, \mathbb{Re})$ be the dual matrix representation of $D(2, \mathbb{Re})$ associated with $\widehat V$ as shown in Proposition~\ref{prop:dual}. We then have $\hat \rho(g)=[\rho(g^{-1})]^\top= \left(
\begin {matrix}
1 & 0 \\
0 & 1/2 \\
\end {matrix}
\right)$ that satisfies $<\hat \rho(g)(\psi), \rho(g)(u)>=<\psi, u>$. Note that $\rho(g), \hat \rho(g) \in D(2, \mathbb{Re})$ for $g \in D(2, \mathbb{Re}$). Now, $\hat \rho(g)$ transforms $\psi=2\hat e_1 + 4\hat e_2$ into $\psi^\prime=2\hat e_1 + (4/2)\hat e_2=2\hat e_1 + 2\hat e_2$. This means that the cost of the term corresponding to $e_2$ has to be reduced to the half of the original cost of the term corresponding to $e_2$ so that the value of $<\psi, u>$ is invariant, i.e., $<\hat \rho(g)(\psi), \rho(g)(u)>=<\psi^\prime, u^\prime>=<\psi, u>$.
\end{example}
\vspace{-10.5pt}

\section{Related Work and Discussion}
\label{sec:relatedwork}
The proper representation of information objects plays an important role in information retrieval (IR), since without it, we cannot expect the good retrieval performance. \\
\indent This paper has assumed that information objects are represented by vectors in a vector space of the VSM and that certain types of transformations of information objects are well-defined by linear transformations of a vector space of the VSM. The results shown in this paper are concerned with the representation of information objects involving several types of transformations in a vector space of the VSM for the purpose of information retrieval (IR), semantics, etc. \\
%\indent Some experiments using linear transformations between vector space bases regarding document representation and ranking have already been conducted in~\cite{Mbarek2014a,Mbarek2014b} but are beyond the scope of this paper. \\
\indent By using group representation theory and linear algebra, this paper provides the mathematical foundation of vector space representation of information objects under group actions, allowing the known group-theoretical results to be adapted for vector space representation of information objects used in IR, semantics, etc.\\
\indent We have discussed several groups of invertible linear transformations on a vector space of the VSM in previous sections.\\
\indent In~\cite{Melucci2005,Melucci2008a,Melucci2008b,Mbarek2014a,Mbarek2014b} context change is modeled by linear transformations from one basis to another in a vector space of the VSM, in order to reflect the information needs evolving with users, time, spaces, etc.\\
\indent Permutation transformations using permutations of vector coordinates on a \emph{word space}~\cite{Sahlgren2008} in order to capture and encode word-order information are discussed in~\cite{Sahlgren2008}.\\
\indent \emph{Unitary transformations}~\cite{Rijsbergen2004,Hoenkamp2003} on a \emph{Hilbert vector space}~\cite{Johnsonbaugh1981} used in IR are discussed in~\cite{Rijsbergen2004}, in which a Hilbert vector space is a \emph{complete inner product space}~\cite{Johnsonbaugh1981}.\\ 
\indent In~\cite{Rijsbergen2004} the notions of quantum mechanics (QM)~\cite{Feynman1965}, such as \emph{state vector}~\cite{Feynman1965}, \emph{observable}~\cite{Feynman1965,Rijsbergen2004}, \emph{superposition}~\cite{Beiser1995}, and \emph{uncertainty}~\cite{Beiser1995,Feynman1965}, are translated into the notions of IR, intending to apply some of the known theorems (e.g., \emph{Gleason's theorem}~\cite{Rijsbergen2004}) of QM to the IR context. In that book a document is represented by a vector in a Hilbert vector space, while relevance is represented by a \emph{Hermitian operator}~\cite{Rijsbergen2004} that encapsulates the uncertainty involving relevance. (The interested reader may also refer to~\cite{Widdows2004} for the geometry of conceptual space using vector spaces and quantum theory.)\\
\indent The dual space model for semantic relations and compositions is discussed in~\cite{Turney2012}, which consists of a domain space and  a function space for two distinct similarity measures. However, it does not involve any dual space of linear functionals on a vector space. Meanwhile, the dual space of linear functionals on a vector space is involved in \emph{Dirac notation}~\cite{Rijsbergen2004,Feynman1965} that is used for~\emph{relevance feedback}~\cite{Grossman2004} and \emph{ostensive retrieval}~\cite{Rijsbergen2004} in IR (see~\cite{Rijsbergen2004}).\\
\indent Although group representation theory involving \emph{tensor product}~\cite{Dummit2004} of vector spaces are well-studied in mathematics~\cite{Fulton1991}, we have not considered any tensor product of vector spaces for semantics in terms of group representation theory in this paper. In computational and mathematical linguistics~\cite{Coecke2010,Turney2010,Turney2012} the vector space tensor product is often used to model \emph{compositionality}~\cite{Mitchell2011} (see~\cite{Grefenstette2011,Grefenstette2013} also for tensor-based compositionality). In~\cite{Clark2008,Coecke2010} a \emph{compositional distributional model of meaning}~\cite{Coecke2010,Clark2008} using \emph{category theory}~\cite{Maclane1998} is discussed, where tensor product is employed for the composition of meanings and types. In that framework VSMs are used for distributional theory of meaning~\cite{Curran2004}, and \emph{Pregroups}~\cite{Clark2008} are used for a compositional theory for grammatical types~\cite{Coecke2010}. (Both the category of vector spaces and the category of Pregroups are examples of \emph{compact closed categories}~\cite{Maclane1998}. The interested reader may refer to~\cite{Coecke2010} for further details.) We leave it as our future work to consider tensor product representation of certain information objects under group actions for the purpose of IR, semantics, machine learning~\cite{Duda2012}, etc.

\section{Conclusions}
\label{sec:Conclusions}
Although group theory is a major area of research in mathematics, few researches have been done how it is utilized for the VSM. This paper discussed certain dynamic transformations of information objects used in the VSM by means of group-theoretical methods. In our framework an information object is considered as a dynamic entity rather than a static one, where a dynamic transformation of information objects is represented by an element of a group of invertible linear operators on a vector space of the VSM. Several groups act on a vector space $V$ of the VSM by means of their matrix representations, in order to perform a dynamic transformation of information-object vectors systematically. We also showed how the dual space $\widehat V$ of $V$ can be employed for the existing VSM. We leave it as an open question to allow other groups that are not discussed in this paper to act on a vector space of the VSM and to derive the useful properties involving some dynamic transformations of information objects used in the VSM.

\begin{appendix}
\section*{Appendix. Vector Spaces, Groups, and Representations}
\label{sec:Appendix}
\setcounter{theorem}{0}
 \renewcommand{\thetheorem}{A.\arabic{theorem}}
\setcounter{lemma}{0}
    \renewcommand{\thelemma}{A.\arabic{lemma}}
In this section we summarize the necessary mathematical background used in this paper. The definitions and results in this section are found in~\cite{Fraleigh1998,Datta1991,Dummit2004,Hungerford1980,Rotman1994,Alperin1995,Sagan2001, Hall2003,Johnsonbaugh1981,Kim2010,Lee2000,Fulton1991,Anton1993,Bretscher1997}. We assume that the reader has some familiarity with linear algebra.

A \emph{group} $(G,\,\cdot\,)$ is a nonempty set $G$, closed under a binary operation $\cdot$, such that the following axioms are satisfied: (i) $(a\cdot b)\cdot c =  a \cdot (b \cdot c)$ for all $a,b,c \in G$, 
(ii) there is a unique element $e\in G$, called the \emph{identity element} of $G$, such that for all $x\in G,~e \cdot x = x \cdot e = x$,
(iii) for each element $a \in G$, there is an element $a^{-1} \in G$ such that $a \cdot a^{-1}= a^{-1} \cdot a = e$. A group \emph{G} is \emph{abelian} if its binary operator $\cdot$ is commutative such that $a \cdot b = b \cdot a$ for all $a,b \in G$.

Let $I_n = \{1, 2,\ldots, n\}$. The group of all bijections $I_n  \rightarrow I_n$, whose binary operation is function composition, is called the \emph{symmetric group on n letters} and is denoted by $S_n$.

Let $G$ be a group and $H$ be a nonempty subset of a group $G$. If $H$ is a group under the restriction to $H$ of the binary operation of $G$, then $H$ is called a \emph{subgroup} of $G$.

Let $(G, \,\cdot\,)$ and $(G', \,\circ\,)$ be groups. A map $\phi:G \rightarrow G'$ is a \emph{homomorphism} if $\phi(x \cdot y)=\phi(x) \circ \phi(y)$ for all $x,y \in G$.

A \emph{ring} is a nonempty set $R$ together with two binary operations $+\,,\, \times:R \times R \rightarrow R$ (called \emph{addition} and \emph{multiplication}) such that: (i) ($R$, +) is an abelian group, (ii) $(a \times b) \times c=a \times (b \times c)$ for all $a, b, c \in R$, (iii) $a \times (b+c)=(a \times b)+(a \times c)$ and $(a+b)\times c=(a \times c) + (b \times c)$. In addition, (iv) if $a \times b = b \times a$ for all $a,b \in R$, then $R$ is said to be a \emph{commutative ring}, (v) if $R$ contains an element $1_R$ such that $1_R \times a= a\times 1_R=a$ for all $a \in R$, then $R$ is said to be a ring with \emph{unity}. 

If $(R,\,+\,,\,\times\,)$ is a ring and $(G,\,\cdot\,)$ is a group, we also write $ab$ rather than $a \times b$ for $a, b \in R$, and write $ab$ rather than $a \cdot b$ for $a, b \in G$, respectively.

An element $x$ in a ring $R$ with unity is said to be \emph{left} (respectively, \emph{right}) \emph{invertible} if there exists an element $z$ (respectively, $y \in R$) in a ring $R$ such that $zx = 1_R$ (respectively, $xy = 1_R$). An element $x \in R$ that is both left and right invertible is said to be a \emph{unit}.

A ring $R$ with unity $1_R \neq 0$ in which every nonzero element is a unit is called a \emph{division ring}. A \emph{field} is a commutative division ring.

Let $R$ be a ring. A \emph{(left) R-module} is an additive abelian group $M$ together with a \emph{scalar multiplication} defined by a function $R \times M \rightarrow M$ such that for all $r, s \in R$ and $a,b \in M$: (i) {$(rs)a = r(sa)$}, (ii) {$(r+s)a = ra + sa$}, (iii) {$r(a+b)=ra + rb$}. In addition, if $R$ is a ring with unity and {$1_Ra = a$ for all $a \in M$, then $M$ is a unitary \emph{R-module}}.

If $R$ is a field, a unitary $R$-module $M$ is called a \emph{vector space} $M$ over $R$.

In the remainder of this paper $G$ denotes a group, $\mathbb{K}$ a field, and $V$ denotes a finite-dimensional vector space unless otherwise stated.

Let $V$, $W$ be vector spaces over $\mathbb{K}$. A function $T:V \rightarrow W$ is a \emph{linear transformation} from $V$ to $W$ provided that for all $x, y \in V$ and $k \in \mathbb{K}$: (i) $T(x+y)=T(x)+T(y)$, (ii) $T(kx)=kT(x)$. A linear transformation from $V$ to itself is also called a \emph{linear operator} of $V$.
  
A (left) \emph{action} of a group $G$ on a set $X$ is a function $G \times X \rightarrow X$ (given by $(g, x)\mapsto gx$) such that for all $x \in X$ and $g_1, g_2 \in G$: (i) $ex = x$, (ii) $(g_1g_2)x =g_1(g_2x)$. When such an action is given, we say that $G$ \emph{acts} (left) on the set $X$.

The \emph{general linear group} $GL(n, \mathbb{K})$ is the group of all invertible $n \times n$ matrices with entries from $\mathbb{K}$ under matrix multiplication. An $n\times n$ matrix is invertible if and only if its determinant is not zero. Alternatively, the general linear group of $V$ is the group of all invertible linear transformations from $V$ to $V$ and is denoted by $GL(V)$. (If $V$ is a finite $n$-dimensional vector space, then $GL(n, \mathbb{K})$ and $GL(V)$ are isomorphic as groups. See~\cite{Dummit2004} for details.)

The general linear group $GL(n, \mathbb{K})$ and its subgroups act on $V=\mathbb{Re}^n$ by matrix multiplication, considering each vector in $V$ as a column matrix. (That is, if $M \in GL(n, \mathbb{K})$ and $x \in V$, $(M, x)\mapsto Mx$.)

A \emph{linear representation} of $G$ is a group homomorphism $\rho:G \rightarrow GL(V)$ from $G$ into $GL(V)$. Similarly, a \emph{matrix representation} of $G$ is a group homomorphism $\rho^\prime:G \rightarrow$ $GL(n, \mathbb{K})$ from $G$ into $GL(n, \mathbb{K})$.

Suppose $G$ acts on a vector space $V$ over $\mathbb{K}$. The action of $G$ on $V$ is called \emph{linear} if the following conditions are met: (i) $g(v+w)=gv+gw$ for all $g \in G$ and $v, w \in V$,
(ii) $g(kv)=k(gv)$ for all $g \in G, k \in \mathbb{K}$, and $v \in V$. If $G$ acts on $V$ linearly, then $V$ itself is called a representation of $G$, and write  $gv$ or $g \cdot v$ for $\rho(g)(v)$ .

Let $V$ be a vector space over $\mathbb{Re}$. An \emph{inner product} for $V$ is a function ( , ) from $V \times V$ into $\mathbb{Re}$ which satisfies the following for all $x,y,z \in V$ and for all  $k \in \mathbb{Re}$: (i) 
$(kx+y, z)=k(x,z)+(y,z)$, (ii) $(x, y)=(y,x)$, (iii) $(x,x) \geq 0$, (iv) if $(x, x)=0$, then $x=0$.

\begin{theorem}[\cite{Johnsonbaugh1981}]
\label{thm:Innerproduct}
The equation
\begin{center}
$(x, y)=\displaystyle\sum_{k=1}^nx_ky_k$,
\end{center}
where $x=(x_1,\ldots, x_n),\, y=(y_1,\ldots, y_n) \in \mathbb{Re}^n$ defines an inner product on $\mathbb{Re}^n$.
\end{theorem}

Let $\|v\|={(\sum_{i=1}^nv_{i}^2)}^\frac{1}{2}$, where $v=(v_1,\ldots, v_n) \in \mathbb{Re}^n$. Then, the geometric interpretation of $(u, v)$ is $(u, v)=\|u\|\|v\|\text{cos }\theta$, where $\theta$ is the angle between $u$ and $v$. For an inner product on $V=\mathbb{Re}^n$, we write $u \cdot v$ rather than $(u, v)$ 

If we change an ordered basis $B=\{b_1,\ldots, b_n\}$ of an $n$-dimensional vector space $V$ to the new ordered basis $B^\prime=\{b_1^\prime,\ldots, b_n^\prime\}$, then a vector $v$ has old coordinate matrix $[v]_B$ and a new coordinate matrix $[v]_{B^\prime}$, respectively. It is related to the equation $[v]_B=S[v]_{B^\prime}$, where $S$ is called the \emph{transition matrix} from $B^\prime$ to $B$. If $X$ and $Y$ are square matrices (i.e., $n \times n$ matrices), then $Y$ is \emph{similar to} $X$ if there is an invertible matrix $P$ such that $Y=P^{-1}XP$.

Let $M=(a_{ij})$ be an $n\times n$ matrix. The \emph{main diagonal} of $M$ consists of the entries $a_{ii}$ for $1\leq i \leq n$. A matrix $D$ is called \emph{diagonal} if its non-zero entries appear only on the main diagonal. A matrix $U$ is called \emph{upper triangular} if all entries of $U$ lying below the main diagonal are zero. A matrix $L$ is called \emph{lower triangular} if all entries of $L$ lying above the main diagonal are zero. 

A square matrix $M$ is called \emph{diagonalizable} if it is similar to a diagonal matrix. 

A linear operator $T$ of $V$ is called \emph{diagonalizable} if there exists an (ordered) basis of $V$ with respect to which the transformation matrix of $T$ is a diagonal matrix. 

A square matrix is called \emph{symmetric} if $A=A^\top$.

Let $V$ be a vector space over $\mathbb{K}$. If $T$ is a linear operator of $V$, a nonzero vector $v \in V$ satisfying $Tv = \lambda v$ for some $\lambda \in \mathbb{K}$ is called an \emph{eigenvector} of $T$. The following theorem describes the fundamental fact of a diagonalizable matrix.

\begin{theorem}[\cite{Anton1993}]
\label{thm:diagonalizable}
An $n\times n$ matrix $M$ with real entries is diagonalizable if and only if $M$ has $n$ linearly independent eigenvectors. 
\end{theorem}

\begin{lemma}[\cite{Anton1993,Bretscher1997}]
\label{lem:symmetricmatrix}
Every symmetric matrix with real entries is diagonalizable. 
\end{lemma}

\begin{lemma}[\cite{Anton1993,Bretscher1997}]
\label{lem:determinants}
Similar matrices have the same determinant. 
\end{lemma}

Given a linear operator $T:V \rightarrow V$, the following theorem describes how the transformation matrix of a linear operator of $V$ changes as we change a basis.
\begin{theorem}[\cite{Anton1993}]
\label{thm:changeofbasis}
Let $T:V \rightarrow V$ be a linear operator of $V$ and let $B$ and $B^\prime$ be both bases for $V$. Then, $[T]_B$ and $[T]_{B^\prime}$ are similar, where $[T]_B$ (respectively, $[T]_{B^\prime}$) denotes the transformation matrix of $T$ with respect to $B$ (respectively, $B^\prime$). Specifically, $[T]_{B^\prime} = S^{-1}[T]_BS$, where $S$ is the transition matrix from $B^\prime$ to $B$.
\end{theorem}

Let $V$ be a vector space over $\mathbb{Re}$, and $\mathbb{Re}$ be a one-dimensional vector space over itself. Let $\text{Hom}_{\mathbb{Re}}(V, \mathbb{Re})$ be the set of all linear transformations from $V$ to $\mathbb{Re}$. This set, denoted by $\widehat V$, forms a vector space over $\mathbb{Re}$, which is called the \emph{dual space} of $V$. Elements of $\widehat V$ are called \emph{linear functionals}. 

\begin{theorem}[\cite{Dummit2004}]
\label{thm:dualbasis}
If $B=\{v_1,\ldots, v_n\}$ is a basis of a vector space V over $\mathbb{Re}$, define $\hat v_i \in \widehat V$ for each $i \in \{1,\ldots,n\}$ by its action on the basis $B$ in such a manner that
$\hat v_i(v_j)=\delta_{ij}$ for $1 \leq j \leq n$, where $\delta_{ij}$ for $1 \leq j \leq n$ denotes $0 \in \mathbb{Re}$ if $i \neq j$ and $1 \in \mathbb{Re}$ if $i=j$. Then, $\widehat V$ is a vector space over $\mathbb{Re}$ with basis $\hat{B}=\{\hat v_1,\ldots, \hat v_n\}$. 
\end{theorem}

There is a (bilinear) natural pairing $<\cdot, \cdot>$ between $\widehat V$ and \emph{V} defined by $<\phi, v> \stackrel{\rm{def}}{=}  \phi(v)$ for $\phi \in \widehat V$ and $v \in V$. (If $A$ denotes a linear operator of $V$ and $A^\top$ denotes its dual or transpose operator of $\widehat V$, $(A^\top\phi)(v)=\phi(Av)$ for $\phi \in \widehat V, v \in V$~\cite{Hall2003}.)

Let $\widehat V=\text{Hom}_{\mathbb{Re}}(V, \mathbb{Re})$ be the dual space of $V$ and let $\rho:G \rightarrow GL(V)$ be a representation of $G$. The \emph{dual representation} $\hat \rho:G \rightarrow GL(\widehat V)$ to $\rho:G \rightarrow GL(V)$ is the representation of $G$ acting on $\widehat V$ given by $\hat \rho(g)=[\rho(g^{-1})]^\top:\widehat V \rightarrow \widehat V$, where $\hat \rho(g)$ is the transpose of $\rho(g^{-1})$. 

The following lemma describe the relationship between a representation $\rho:G \rightarrow GL(V)$ of $G$ and the dual representation $\hat \rho:G \rightarrow GL(\widehat V)$ of $G$.

\begin{lemma}[\cite{Fulton1991}]
\label{lem:dualrep}
$<\hat \rho(g)(\hat v), \rho(g)(v)>=<\hat v, v>$ for all $g \in G, v \in V$, and $\hat v \in \widehat V$.
\end{lemma}

The definition of the dual representation is such that the following diagram commutes~\cite{Fulton1991}:
\begin{center}
\[
\begin{CD}
V@> \phi >> \mathbb{Re} \\
@V{g} VV @ VV {g} V \\
V @>> {g\phi} > \mathbb{Re}\\
\end{CD}
\]
\end{center}
Therefore, $(g\phi)(v)=g\phi(g^{-1}v)$ for all $g \in G$ and $v \in V$. Since $gx=x$ for all $x \in \mathbb{Re}$, we have $(g\phi)(v)=g\phi(g^{-1}v)=\phi(g^{-1}v)$. Since $\phi(g^{-1}v)=({(g^{-1})^\top}\phi)(v)$, we have $g\phi={(g^{-1})^\top}\phi$, which corresponds to the above definition.
\end{appendix}
\bibliographystyle{gCOM}
\bibliography{dkim}
\end{document}